\documentclass[
reprint,
superscriptaddress,
nobibnotes,
amsmath,amssymb,
aps,
prb,
]{revtex4-2}

\usepackage{bm}
\usepackage{amsmath}
\usepackage{amssymb}
\usepackage{amsthm}
\usepackage{graphicx}
\usepackage{booktabs}
\usepackage[hypertexnames=false]{hyperref}

\newtheorem{theorem}{Theorem}

\begin{document}

\title{Thermodynamic incompleteness in non-Markovian Majorana transport}

\author{Yang Tian}
\email{tyanyang04@gmail.com \& yang.tian@infplane.com}
\affiliation{Infplane Computing Technologies Ltd, Beijing, 100080, China}

\begin{abstract}
We show that the complete knowledge of the non-Markovian island-state dynamics of
a floating Majorana island does not, in general, determine the thermodynamic
transport statistics measured in the leads. We demonstrate this statement in
a Coulomb-blockaded island with $M$ Majorana zero modes coupled to structured
reservoirs. In the cotunneling regime, a Schrieffer-Wolff transformation gives
reservoir-assisted transitions generated by Majorana bilinears. After the
reservoirs are traced out, the island state determines the memory kernel
associated with each bilinear, and this is enough to predict all island-state
observables within the cotunneling approximation. It is not enough to
determine which lead or detector channel supplied the electron, absorbed the
electron, or carried the corresponding energy exchange. This is a genuine
loss of thermodynamic information, not an error in the island equation. We
formulate the result as a thermodynamic completeness criterion: an island
memory equation determines a transport observable only when that observable is
constant over all assignments of reservoir channels that give the same island
memory kernel. The criterion gives a measurable prediction. Two
structured-reservoir Majorana devices can have identical island-state
tomography and relaxation, but different charge noise measured separately in
the leads, heat noise, and mixed charge-energy correlations. The geometry of the projection
from reservoir records to island kernels and the topology of the network of
tunnel contacts identify which transport information is absent from
island-state dynamics.
\end{abstract}

\maketitle

\section{Introduction}

Majorana islands in the Coulomb blockade regime provide a controlled setting
for nonlocal transport, electron teleportation, current statistics, and
topological Kondo physics
\cite{MichaeliLandauSelaFu2017,BeriCooper2012,LutchynGlazman2017,KleinherbersSchuenemannKoenig2023,SchurayRammlerRecher2020}.
Recent transport and full-counting-statistics studies have used charge,
noise, and counting singularities to probe Majorana devices and topological
phases
\cite{ZhangFull2021,WangDistinguishing2024,KleinherbersSchuenemannKoenig2023,SchurayRammlerRecher2020,McCullochDeNardisGopalakrishnanVasseur2023,TirritoFull2023,PaulinoDeviation2024,GerryFull2023,EngelhardtLuoBastidasPlatero2024a,EngelhardtLuoBastidasPlatero2024b}.
Recent work on non-Markovian quantum stochastic processes, structured
reservoirs, and non-Markovian thermodynamics has also clarified how memory
modifies open-system dynamics and entropy production
\cite{MilzModi2021,LuoLambertLiangCirio2023,CockrellMarkovian2022,SenyasaEntropy2022,GhoshalHeat2022,WisniewskiMemory2024,FiorelliGherardiniMarcantoni2023,ZhangWangZengWang2022,WuMancinoCarlessoCiampiniMagriniKieselPaternostro2022}.
These developments leave open a reconstruction problem that is directly
relevant for transport experiments. Charge transfer, heat transfer, and
current noise are records of how electrons enter and leave the reservoirs,
whereas the island density matrix retains only the memory kernel that acts on
the Majorana degrees of freedom.

The reconstruction problem is also connected to several recent directions in
nonequilibrium statistical physics. Quantum and classical thermodynamic
uncertainty relations, speed limits, and response bounds relate current
fluctuations to dissipation and dynamical activity
\cite{HasegawaContinuous2020,HasegawaThermodynamic2021,HasegawaThermodynamic2022,SalazarThermodynamic2024,LiuCoherences2021,MonnaiThermodynamic2022,ReicheThermodynamic2022,ZiyinUniversal2023,HasegawaIrreversibility2021,MonnaiArbitrary2023,AlmanzaMarreroCertifying2025,HamazakiSpeed2022,PoggiDiverging2021,LanGeometric2022}.
Thermodynamic inference, stochastic thermodynamics, and information-geometric
approaches ask which observables can be inferred from partial dynamical data
\cite{LandiPaternostro2021,FalascoEsposito2025,ChatzittofiGolestanianAgudoCanalejo2024,HarunariFioreBarato2024,BettmannInformation2025,KobayashiHessian2022,BrandnerSaito2020,NakazatoGeometrical2021,GiordanoEntropy2021,LeightonJensen2024,TasnimMultiple2023,ChirikjianRate2021,Zamparoquadratic2023,Budhirajaentropy2021,RobertsHidden2021,BrasilThermodynamic2023}.
Large-deviation and macroscopic-fluctuation approaches give a complementary
language in which state variables and currents are kept as separate dynamical
data
\cite{KraaijLazarescuMaesPeletier2020,PattersonRengerSharma2024,RengerSharma2023,DoyonBallistic2023,KunduBallistic2025,AgranovMacroscopic2023,MallickExact2024,KrajenbrinkCrossover2023,BouleyDynamical2024,PeletierGradient2023,MartinOByrneCatesFodorNardiniTailleurVanWijland2021}.
The present paper uses a Majorana island to turn this general distinction into
a concrete condensed-matter statement. The island memory kernel specifies
how the past island density matrix influences the later island density matrix
after the reservoirs have been traced out. In obtaining this kernel, the lead
and channel labels of cotunneling events are summed. A transport measurement
asks a different question: which reservoir channel gained or lost charge and
energy in each event. The same island memory kernel can therefore be
compatible with different counting-field extensions and different measured
transport statistics.

The question addressed here is motivated by our previous work
\cite{Tian2026ThermodynamicCompleteness}, where we show that a Markovian
state generator can be insufficient to determine thermodynamic transport
records. In that Markovian setting, a transport observable is reconstructible
from state dynamics only when it is invariant under all reservoir-channel
redistributions that leave the state generator unchanged. We now extend this
criterion to a non-Markovian Majorana device. The main finding is that
complete knowledge of
non-Markovian island-state dynamics does not, in general, determine the
thermodynamic transport statistics measured in the leads. More precisely, the
island-state memory
equation can determine all island-state observables while still leaving the
transport record measured in individual leads undetermined. In the Coulomb
blockade valley,
non-Markovian island dynamics
depends on sums over lead and channel memory kernels associated with Majorana
bilinears, whereas charge, heat, and noise statistics depend on the individual
lead and channel processes before those sums are taken. We first derive the
cotunneling Hamiltonian for an $M$-terminal floating island. We then derive
the non-Markovian memory kernels generated by structured reservoirs and
introduce counting fields before the lead and channel labels are summed out.
The central theorem shows that island-state dynamics fixes only the
kernels associated with each Majorana bilinear, not the individual lead and
channel processes that produce those kernels. We finally express this loss of
information as a projection in memory-kernel space and show how the
topology of the network of tunnel contacts and the fixed fermion-parity
constraints of Majorana zero modes determine which transport records are
absent from island state dynamics.

\section{Floating island and cotunneling Hamiltonian}
\label{Model-section}

We begin with a microscopic model that keeps the island operators and the
reservoir channels separate. This separation is essential because the
island-state dynamics will later sum over channel labels, while thermodynamic
measurements keep them as part of the transport record.

We consider a floating superconducting island with charging Hamiltonian
\begin{equation}
H_{C}
=E_{C}\left(\hat{N}-n_{g}\right)^{2}.
\label{II-EQ001}
\end{equation}
The island hosts $M$ spatially separated Majorana zero modes
$\gamma_{j}=\gamma_{j}^{\dagger}$, with
\begin{equation}
\left\{\gamma_{j},\gamma_{k}\right\}
=2\delta_{jk}.
\label{II-EQ002}
\end{equation}
Contact $j$ is coupled to a structured set of reservoir channels $r$. The
tunneling Hamiltonian is
\begin{equation}
H_{T}
=\sum_{j=1}^{M}\sum_{r}
\left[
t_{jr}\psi_{jr}^{\dagger}\gamma_{j}e^{-i\hat{\phi}/2}
+t_{jr}^{\ast}e^{i\hat{\phi}/2}\gamma_{j}\psi_{jr}
\right],
\label{II-EQ003}
\end{equation}
where $e^{\mp i\hat{\phi}/2}$ changes the island charge by one electron. The
reservoir channels may represent distinct leads, energy filters, spin filters,
or calorimetric detector channels. The channel index is therefore part of the
measured transport record, not a relabeling introduced after the measurement.

We work in a Coulomb blockade valley with fixed charge $N_{0}$. In this
regime, the island cannot change its charge in a real process, but virtual
charge states still mediate cotunneling. The corresponding virtual excitation
energies are
\begin{align}
E_{+}
&=E_{C}\left(N_{0}+1-n_{g}\right)^{2}
-E_{C}\left(N_{0}-n_{g}\right)^{2},
\label{II-EQ004}\\
E_{-}
&=E_{C}\left(N_{0}-1-n_{g}\right)^{2}
-E_{C}\left(N_{0}-n_{g}\right)^{2}.
\label{II-EQ005}
\end{align}
These two energy costs set the strength of the low-energy cotunneling
processes. To second order in $H_{T}$ and within the fixed-charge manifold
$N=N_{0}$, a standard Schrieffer-Wolff transformation
\cite{SchriefferWolff1966} gives the cotunneling Hamiltonian
\begin{equation}
H_{\rm cot}
=\sum_{j<k}\sum_{r,s}
\left[
J_{jr,ks}\psi_{jr}^{\dagger}\psi_{ks}S_{jk}
+J_{jr,ks}^{\ast}S_{jk}\psi_{ks}^{\dagger}\psi_{jr}
\right],
\label{II-EQ006}
\end{equation}
where the island operator carried by a cotunneling event is the Majorana
bilinear
\begin{equation}
S_{jk}=i\gamma_{j}\gamma_{k}
\label{II-EQ007}
\end{equation}
and the cotunneling amplitude is
\begin{equation}
J_{jr,ks}
=\frac{t_{jr}t_{ks}^{\ast}}{2}
\left(\frac{1}{E_{+}}+\frac{1}{E_{-}}\right).
\label{II-EQ008}
\end{equation}
We obtain Eq. (\ref{II-EQ006}) by projecting the tunneling Hamiltonian onto
the fixed-charge manifold and keeping the two virtual paths in which the
island charge is first raised or first lowered. The Schrieffer-Wolff
expression
$H_{\rm eff}=-PH_{T}Q\left(QH_{C}Q-E_{0}\right)^{-1}QH_{T}P$ gives the two
energy denominators $E_{+}^{-1}$ and $E_{-}^{-1}$. The Majorana product
generated by a process from contact $k$ to contact $j$ is
$\gamma_{j}\gamma_{k}=-iS_{jk}$, which gives Eqs.
(\ref{II-EQ006}-\ref{II-EQ008}) after collecting the Hermitian-conjugate
process.

Potential-scattering terms with $j=k$ do not change the island state and are
not needed for the island-state equation considered below. If such terms are
included in transport counting, they add further reservoir records that are
not visible in the island state and therefore strengthen, rather than weaken,
the incompleteness result. The cotunneling Hamiltonian separates the two
pieces of information carried by one event: the island sees the bilinear
$S_{jk}$, whereas the reservoirs record the channel pair
$\left(jr,ks\right)$. The operators $S_{jk}$ satisfy the
$\mathfrak{so}\left(M\right)$ commutation relations
\begin{align}
\left[S_{jk},S_{\ell m}\right]
=2i(&\delta_{k\ell}S_{jm}-\delta_{j\ell}S_{km}
\nonumber\\
&-\delta_{km}S_{j\ell}+\delta_{jm}S_{k\ell}).
\label{II-EQ009}
\end{align}
which encode the nonlocal Majorana algebra that underlies the island-state
dynamics. Having identified the event operator seen by the island and the
channel pair recorded by the reservoirs, we now trace out the reservoirs and
derive the memory kernels that enter the non-Markovian island equation.

\section{Non-Markovian memory kernels}
\label{Memory-section}

The cotunneling Hamiltonian contains more information than the island
state can retain. We now show explicitly which part survives after the
structured reservoirs are integrated out.

Let
\begin{equation}
B_{jr,ks}=\psi_{jr}^{\dagger}\psi_{ks}
\label{III-EQ001}
\end{equation}
be the reservoir operator that transfers an electron from channel $ks$ to
channel $jr$. The structured reservoir correlation functions are
\begin{align}
C_{jr,ks}^{>}\left(t\right)
&=\left\langle B_{jr,ks}\left(t\right)
B_{jr,ks}^{\dagger}\left(0\right)\right\rangle,
\label{III-EQ002}\\
C_{jr,ks}^{<}\left(t\right)
&=\left\langle B_{jr,ks}^{\dagger}\left(0\right)
B_{jr,ks}\left(t\right)\right\rangle.
\label{III-EQ003}
\end{align}
For finite-bandwidth reservoirs these functions are not proportional to
delta functions in time. The cotunneling kernels with explicit channel labels
are
\begin{align}
K_{jr,ks}^{>}\left(t\right)
&=\left|J_{jr,ks}\right|^{2}C_{jr,ks}^{>}\left(t\right),
\label{III-EQ004}\\
K_{jr,ks}^{<}\left(t\right)
&=\left|J_{jr,ks}\right|^{2}C_{jr,ks}^{<}\left(t\right).
\label{III-EQ005}
\end{align}
The island-state dynamics does not retain the individual channel pair. It
retains only the kernels labelled by the Majorana bilinear,
\begin{align}
K_{jk}^{>}\left(t\right)
&=\sum_{r,s}K_{jr,ks}^{>}\left(t\right),
\label{III-EQ006}\\
K_{jk}^{<}\left(t\right)
&=\sum_{r,s}K_{jr,ks}^{<}\left(t\right).
\label{III-EQ007}
\end{align}
The second-order island equation follows by inserting these kernels into the
Born memory expansion. In the interaction picture of the isolated island
ground-state manifold, the part that controls the dissipative island-state
dynamics is
\begin{align}
\frac{d\rho\left(t\right)}{dt}
=\sum_{j<k}\int_{0}^{t}d\tau\,
\Gamma_{jk}\left(\tau\right)
\left[
S_{jk}\rho\left(t-\tau\right)S_{jk}
-\rho\left(t-\tau\right)
\right],
\label{III-EQ008}
\end{align}
where
\begin{equation}
\Gamma_{jk}\left(t\right)
=K_{jk}^{>}\left(t\right)+K_{jk}^{<}\left(t\right).
\label{III-EQ009}
\end{equation}
Equation (\ref{III-EQ008}) assumes weak cotunneling, Gaussian reservoirs,
diagonal channel-pair correlations, a fixed island charge, and a fixed total
fermion parity in a degenerate Majorana manifold. Cross-correlated reservoir
channels add off-diagonal memory kernels with the same projection structure.
Principal-value terms can add coherent memory corrections. They are also
constructed from the sums over channel pairs at fixed Majorana bilinear in Eqs.
(\ref{III-EQ006}-\ref{III-EQ007}). In the exactly degenerate Majorana
manifold, the simplest bilinear contribution reduces to a constant and does
not affect the island state. The completeness question is therefore already fixed by the
dissipative kernels in Eq. (\ref{III-EQ008}): the island-state dynamics uses
only the sums over the measured lead and channel labels. The next step is to
insert counting fields before this summation, because transport experiments
measure the channel record that the island-state equation has discarded.

\section{Counting kernels and completeness theorem}
\label{Completeness-section}

The previous section shows where information is lost in the island-state
equation. We now keep the lead and channel labels long enough to define the
transport record. We attach a charge counting field $\chi_{jr}$ and an energy
counting field $\xi_{jr}$ to each measured reservoir channel. For compactness,
we write a discrete channel energy $\varepsilon_{jr}$. For a continuous
structured reservoir, the sums over $r$ are replaced by integrals and the
counting factors are inserted under the energy integral. The pair
$\lambda_{jr}=\left(\chi_{jr},\xi_{jr}\right)$ tilts the cotunneling kernels as
\begin{align}
K_{jr,ks}^{>,\lambda}\left(t\right)
&=K_{jr,ks}^{>}\left(t\right)
e^{i\left(\chi_{jr}-\chi_{ks}\right)
+i\left(\xi_{jr}\varepsilon_{jr}-\xi_{ks}\varepsilon_{ks}\right)},
\label{IV-EQ001}\\
K_{jr,ks}^{<,\lambda}\left(t\right)
&=K_{jr,ks}^{<}\left(t\right)
e^{-i\left(\chi_{jr}-\chi_{ks}\right)
-i\left(\xi_{jr}\varepsilon_{jr}-\xi_{ks}\varepsilon_{ks}\right)}.
\label{IV-EQ002}
\end{align}
Heat counting is obtained by replacing $\varepsilon_{jr}$ with
$\varepsilon_{jr}-\mu_{jr}$ for a reservoir channel with chemical potential
$\mu_{jr}$. The tilted memory kernel for island operator $S_{jk}$ is therefore
\begin{equation}
K_{jk}^{\lambda}\left(t\right)
=\sum_{r,s}
\left[
K_{jr,ks}^{>,\lambda}\left(t\right)
+K_{jr,ks}^{<,\lambda}\left(t\right)
\right].
\label{IV-EQ003}
\end{equation}
The cumulant-generating functional $W\left[\lambda\right]$ is defined by the
tilted non-Markovian evolution generated by Eq. (\ref{IV-EQ003}). Equivalently,
in Laplace space the tilted propagator obeys
\begin{equation}
\rho_{\lambda}\left(z\right)
=\left[
z-L_{\lambda}\left(z\right)
\right]^{-1}\rho\left(0\right),
\label{IV-EQ004}
\end{equation}
where $L_{\lambda}\left(z\right)$ denotes the Laplace transform of the tilted
memory kernel. The dominant singularity of Eq. (\ref{IV-EQ004}) gives the
long-time cumulant-generating function, which generates charge, energy, heat,
and mixed cumulants through
\begin{equation}
\left\langle\!\left\langle I^{\alpha_{1}}_{a_{1}}\cdots I^{\alpha_{n}}_{a_{n}}
\right\rangle\!\right\rangle
=\left.
\frac{\delta^{n}W\left[\lambda\right]}
{\delta\left(i\lambda^{\alpha_{1}}_{a_{1}}\right)\cdots
\delta\left(i\lambda^{\alpha_{n}}_{a_{n}}\right)}
\right|_{\lambda=0},
\label{IV-EQ005}
\end{equation}
where $a=\left(jr\right)$ denotes a measured reservoir channel,
$\alpha\in\{N,E\}$, $\lambda^{N}_{a}=\chi_{a}$, and
$\lambda^{E}_{a}=\xi_{a}$. Eqs. (\ref{IV-EQ001}-\ref{IV-EQ005}) show that the
charge and energy statistics depend on the lead and channel labels before the
sums in Eqs.
(\ref{III-EQ006}-\ref{III-EQ007}) are performed. This observation is the
input to the main thermodynamic completeness theorem.

\begin{theorem}
\label{THM-Incompleteness}
Consider two structured-reservoir realizations of the same floating Majorana
island in the cotunneling regime. If
\begin{align}
\sum_{r,s}K_{jr,ks}^{>,A}\left(t\right)
&=\sum_{r,s}K_{jr,ks}^{>,B}\left(t\right),
\label{IV-EQ006}\\
\sum_{r,s}K_{jr,ks}^{<,A}\left(t\right)
&=\sum_{r,s}K_{jr,ks}^{<,B}\left(t\right)
\label{IV-EQ007}
\end{align}
for every Majorana pair $j<k$ and all memory times $t$, then the two
realizations have identical non-Markovian island-state dynamics to
cotunneling order. If their tilted generating functionals $W^{A}[\lambda]$
and $W^{B}[\lambda]$ are analytic near $\lambda=0$ and differ there, then at
least one charge, energy, heat, or mixed cumulant tied to a specified lead and
reservoir channel differs. Thus the non-Markovian island-state dynamics is not
complete for reconstructing that transport record.
\end{theorem}

\begin{proof}
Eqs. (\ref{III-EQ008}-\ref{III-EQ009}) show that the island-state
dynamics depends on the reservoirs only through the sums over channel pairs
at fixed Majorana bilinear in
Eqs. (\ref{III-EQ006}-\ref{III-EQ007}). Conditions
(\ref{IV-EQ006}) and (\ref{IV-EQ007}) therefore make the island memory
equation identical in the two realizations. The lead-specific cumulants are
derivatives of $W\left[\lambda\right]$. If
$W^{A}\left[\lambda\right]-W^{B}\left[\lambda\right]$ is not identically zero,
its Taylor expansion around $\lambda=0$ has a nonzero coefficient at some
order. The corresponding derivative in Eq. (\ref{IV-EQ005}) is a charge,
energy, heat, or mixed cumulant that differs between the two realizations.
\end{proof}

The kernel of the projection can be constructed explicitly. For a fixed
Majorana pair $j<k$, choose channel-pair variations
$\delta K_{jr,ks}^{>,<}\left(t\right)$ satisfying
\begin{equation}
\sum_{r,s}\delta K_{jr,ks}^{>,<}\left(t\right)=0.
\label{IV-EQ008}
\end{equation}
Then the projected kernels $K_{jk}^{>,<}$ and the island dynamics are
unchanged. The tilted variation is
\begin{align}
\delta K_{jk}^{\lambda}\left(t\right)
=\sum_{r,s}&
\delta K_{jr,ks}^{>}\left(t\right)
e^{i\left(\chi_{jr}-\chi_{ks}\right)
+i\left(\xi_{jr}\varepsilon_{jr}-\xi_{ks}\varepsilon_{ks}\right)}
\nonumber\\
&+\delta K_{jr,ks}^{<}\left(t\right)
e^{-i\left(\chi_{jr}-\chi_{ks}\right)
-i\left(\xi_{jr}\varepsilon_{jr}-\xi_{ks}\varepsilon_{ks}\right)}.
\label{IV-EQ009}
\end{align}
For generic channel variations, Eq. (\ref{IV-EQ009}) is not zero. The current
record changes although the non-Markovian island-state dynamics does not.
This proves that complete island-state dynamics is thermodynamically
incomplete in the cotunneling regime.

The criterion gives a direct experimental prediction. Consider one Majorana
bilinear $S_{jk}$ and two energy-filtered channels at each contact. Let
\begin{equation}
K_{jr,ks}^{>}\left(t\right)
+K_{jr,ks}^{<}\left(t\right)
=R_{rs}f\left(t\right),
\quad
\int_{0}^{\infty}dt\,f\left(t\right)=1,
\label{IV-EQ010}
\end{equation}
where $f\left(t\right)$ is the same structured-reservoir memory profile for
the four channel pairs. Such matching can be engineered, for example, by using
energy filters with the same bandwidth and tuning tunnel amplitudes to set the
weights $R_{rs}$. The integrated cotunneling weight for a channel pair is then
\begin{equation}
R_{rs}
=\int_{0}^{\infty}dt\,
\left[
K_{jr,ks}^{>}\left(t\right)
+K_{jr,ks}^{<}\left(t\right)
\right]
\label{IV-EQ011}
\end{equation}
and the island memory kernel generated by this bilinear depends only on
\begin{equation}
R=\sum_{r,s}R_{rs}.
\label{IV-EQ012}
\end{equation}
The leading zero-frequency contribution of these weak cotunneling events to
the charge and energy noise, however, keeps the individual channel-pair
weights:
\begin{align}
S^{NN}_{ab}
&=\sum_{r,s}R_{rs}\,q^{rs}_{a}q^{rs}_{b},
\label{IV-EQ013}\\
S^{EE}_{ab}
&=\sum_{r,s}R_{rs}\,\epsilon^{rs}_{a}\epsilon^{rs}_{b},
\label{IV-EQ014}\\
S^{NE}_{ab}
&=\sum_{r,s}R_{rs}\,q^{rs}_{a}\epsilon^{rs}_{b}.
\label{IV-EQ015}
\end{align}
Here $q^{rs}_{a}=\delta_{a,jr}-\delta_{a,ks}$ is the charge increment in
measured channel $a$, and
$\epsilon^{rs}_{a}=\varepsilon_{jr}\delta_{a,jr}
-\varepsilon_{ks}\delta_{a,ks}$ is the corresponding energy increment. Now
compare two devices with the same total weight $R$. In the first device,
which we denote by $A$, the nonzero channel-pair weights are the diagonal
ones:
\begin{equation}
R^{A}_{11}=R^{A}_{22}=\frac{R}{2},
\quad
R^{A}_{12}=R^{A}_{21}=0.
\label{IV-EQ016}
\end{equation}
In the second device, denoted by $B$, the nonzero channel-pair weights are
the off-diagonal ones:
\begin{equation}
R^{B}_{12}=R^{B}_{21}=\frac{R}{2},
\quad
R^{B}_{11}=R^{B}_{22}=0.
\label{IV-EQ017}
\end{equation}
They have the same island memory kernel and therefore the same island-state
tomography in the cotunneling regime. Their heat-noise cross correlation
between the two contacts differs by
\begin{equation}
S^{EE,A}_{j,k}-S^{EE,B}_{j,k}
=-\frac{R}{2}
\left(\varepsilon_{j1}-\varepsilon_{j2}\right)
\left(\varepsilon_{k1}-\varepsilon_{k2}\right),
\label{IV-EQ018}
\end{equation}
where $S^{EE}_{j,k}$ denotes the sum of $S^{EE}_{jr,ks}$ over the two
channels at each contact. Thus energy-filtered noise can distinguish two
Majorana islands that are indistinguishable by island-state relaxation. We now
recast the same result geometrically, which makes clear which classes of
transport records are lost under projection.

\section{Geometry and topology of the projection}
\label{Geometry-section}

Theorem \ref{THM-Incompleteness} identifies the missing data algebraically. We
now give the same statement a geometric form. This form is useful because it
separates two questions that are often blended together: which variables are
needed to evolve the island state, and which variables are needed to
reconstruct thermodynamic transport observables. Let
$V_{\rm rec}$ be the vector space of kernels with explicit lead and channel
labels
$K_{jr,ks}^{>,<}\left(t\right)$, and let $V_{\rm op}$ be the vector space of
kernels $K_{jk}^{>,<}\left(t\right)$ labelled only by the Majorana bilinear.
The projection
\begin{equation}
P:V_{\rm rec}\rightarrow V_{\rm op}
\label{V-EQ001}
\end{equation}
is
\begin{equation}
\left(PK\right)_{jk}^{>,<}\left(t\right)
=\sum_{r,s}K_{jr,ks}^{>,<}\left(t\right).
\label{V-EQ002}
\end{equation}
The island-state dynamics depends on $PK$. The thermodynamic record depends
on the full $K$. Therefore, the variations of individual lead and channel
kernels that are not represented in the island-state dynamics are
\begin{equation}
\begin{aligned}
{\rm ker}\,P
=\Big\{
\delta K\in V_{\rm rec}:\,
&\sum_{r,s}\delta K_{jr,ks}^{>,<}\left(t\right)=0
\\
&{\rm for\ all}\ j<k,t
\Big\}.
\end{aligned}
\label{V-EQ003}
\end{equation}
\begin{theorem}
\label{THM-Projection}
A transport observable $O[K]$ is reconstructible from the non-Markovian
island-state dynamics if and only if it is constant on every fiber of $P$:
\begin{equation}
PK^{A}=PK^{B}
\quad \Longrightarrow \quad
O[K^{A}]=O[K^{B}].
\label{V-EQ004}
\end{equation}
If $O$ is differentiable, Eq. (\ref{V-EQ004}) implies
\begin{equation}
\delta O\left[\delta K\right]=0
\quad{\rm for\ all}\quad
\delta K\in{\rm ker}\,P.
\label{V-EQ005}
\end{equation}
\end{theorem}

\begin{proof}
If $O$ can be reconstructed from island-state dynamics, then there exists a
functional $\bar{O}$ on projected kernels such that $O[K]=\bar{O}[PK]$. Hence
two kernels with the same $PK$ give the same $O$, which proves the necessity
of Eq. (\ref{V-EQ004}). Conversely, if Eq. (\ref{V-EQ004}) holds, then
$O[K]$ has a unique value on each fiber of $P$, so $\bar{O}[PK]$ is
well-defined by choosing any representative $K$ in that fiber. Differentiating
Eq. (\ref{V-EQ004}) along a path $K+\epsilon\delta K$ with
$\delta K\in{\rm ker}\,P$ gives Eq. (\ref{V-EQ005}).
\end{proof}

Eq. (\ref{V-EQ005}) is the completeness criterion in differential form:
the gradient of a reconstructible observable must annihilate all variations of
individual lead and channel kernels removed by the projection. This is the
geometric
mechanism behind thermodynamic incompleteness in this paper: island-state
dynamics is complete only for observables whose gradients have no component
along ${\rm ker}\,P$.

The topology of the network of tunnel contacts explains which elements of
${\rm ker}\,P$ are enforced by the arrangement of contacts and cotunneling
paths, rather than by a special choice of tunnel amplitudes. We represent the
contacts by vertices. A possible cotunneling process between contacts $j$ and
$k$ defines an edge carrying the island operator $S_{jk}$. Different
reservoir-channel pairs $\left(r,s\right)$ above the same contact pair
$\left(j,k\right)$ are distinct measured transport processes, but they all
project to the same island operator. The map $P$ therefore removes
redistributions among these channel pairs and places them inside
${\rm ker}\,P$.

This description in terms of the network of contacts has a direct transport
meaning. Island-state
dynamics can detect which Majorana bilinear has acted on the island. It
cannot detect a circulation of current among reservoir channels if that
circulation leaves the projected island kernel unchanged. If we project
measured currents further to net charge accumulation at the contacts, the
relevant linear map is the incidence map
\begin{equation}
\partial:C_{1}\rightarrow C_{0}
\label{V-EQ006}
\end{equation}
from oriented cotunneling edges $C_{1}$ to contact charges $C_{0}$. Currents
in the cycle space of this map have zero net boundary charge. They can
therefore change current noise and higher cumulants while leaving both the
contact charge accumulation and the island-state dynamics unchanged. These
cycle currents are concrete examples of variations in ${\rm ker}\,P$ and
therefore fall under Theorem \ref{THM-Projection}. The two energy-filtered
devices in Eqs. (\ref{IV-EQ016}) and (\ref{IV-EQ017}) realize the same idea in
the smallest channel space: diagonal and off-diagonal channel pairings have
the same projection onto the island memory kernel, but they occupy different
directions in the space of measured transport records and therefore give
different heat noise in Eq. (\ref{IV-EQ018}).

The fixed total fermion parity of the Majorana island gives an additional
constraint. In a four-Majorana subspace, for example,
\begin{equation}
S_{12}S_{34}
=-P_{4},
\quad
P_{4}=\gamma_{1}\gamma_{2}\gamma_{3}\gamma_{4}.
\label{V-EQ007}
\end{equation}
Fixing the fermion parity of the island fixes $P_{4}$, so complementary
cotunneling edges are not independent as island operators even though they are
distinct transport records. For general $M$, products of Majorana bilinears
satisfy analogous fixed fermion-parity relations. Thus there are two sources of
transport information that can be absent from island-state dynamics. The
topology of the network of tunnel contacts creates cycle currents and
redistributions among reservoir channels that do not appear in the projected
island memory kernel. The fixed fermion-parity constraint identifies distinct
cotunneling records that act in the same way on the island within a fixed fermion-parity
subspace. If a transport statistic distinguishes these records, its variation
along ${\rm ker}\,P$ is nonzero and the criterion in Eq. (\ref{V-EQ004})
fails. Thus the topology of the contacts and the fixed fermion-parity
constraint specify which transport information is absent from island-state
dynamics. The discussion below summarizes why this obstruction changes the
interpretation of non-Markovian island dynamics in Majorana transport.

\section{Discussion}
\label{Discussion-section}

Recent studies of Majorana transport and full counting statistics usually
start from a specified device and a specified measurement protocol. The
counting fields are attached to the chosen leads or channels from the outset
\cite{ZhangFull2021,KleinherbersSchuenemannKoenig2023,WangDistinguishing2024,SchurayRammlerRecher2020}.
By contrast, recent non-Markovian thermodynamic studies often ask how
reservoir memory changes open-system dynamics, entropy production, or heat
currents
\cite{MilzModi2021,CockrellMarkovian2022,SenyasaEntropy2022,GhoshalHeat2022,LuoLambertLiangCirio2023}.
The present result connects these two viewpoints by asking whether complete
knowledge of the non-Markovian island dynamics determines the counting
construction needed for transport thermodynamics. The answer is no. The
island dynamics determines
only the kernels $K_{jk}^{>,<}$ associated with Majorana bilinears. It does not
determine the individual lead and reservoir-channel cotunneling processes
$K_{jr,ks}^{>,<}$ that set charge-current means, heat currents, current noise,
and higher cumulants.

This conclusion rules out a common but misleading identification between
complete state dynamics and complete thermodynamic information. A memory-kernel
equation can predict every island density matrix in its regime of validity and
still omit the
information needed to reconstruct the measured transport statistics. The
missing information is not a small correction to the island dynamics. It is
the reservoir record of which channel participates in each cotunneling event.
For Majorana islands, this distinction is especially sharp because the island
sees the nonlocal bilinear $S_{jk}$, whereas the leads distinguish the many
possible reservoir-channel realizations of that same bilinear event.

The result extends the Markovian criterion of our previous work
\cite{Tian2026ThermodynamicCompleteness} from reservoir-channel assignments in
ordinary Markovian transport to memory-kernel decompositions in a
non-Markovian Majorana island. The new feature is that the missing
thermodynamic information is organized by Majorana bilinears,
structured-reservoir memory, and the topology of the network of tunnel
contacts. The geometric kernel ${\rm ker}\,P$ contains the variations of
individual lead and channel kernels lost from the island-state description,
while the topology of the network of tunnel contacts and fixed fermion-parity
constraints identify which of those variations can still be measured as
transport records. This gives a practical diagnostic for modeling: island-state
tomography or fitting a memory kernel is not enough to predict heat, charge,
and noise unless the counting-field extension is also fixed.

The experimentally relevant implication is that two structured-reservoir
Majorana devices may have the same island-state dynamics while current noise
or higher cumulants distinguish them. Such a distinction is not merely a
matter of detector resolution. It reflects the
fact that reservoir-channel information is projected out when one constructs
the island memory equation. For Majorana transport, this means that nonlocal
cotunneling signatures can reside in correlations between leads even when the
island memory kernel is unchanged. Measurements of cross noise and
higher-order counting statistics therefore test information that is absent
from island-state dynamics.

The controlled regime of this paper is the weak-cotunneling Coulomb blockade
valley. At still lower energies, the cotunneling couplings may flow toward a
topological Kondo fixed point \cite{BeriCooper2012,LutchynGlazman2017}. A
complete treatment of that regime requires a boundary counting action with
explicit lead labels rather than the perturbative kernels used here. The same
reconstruction question can then be asked for boundary current records at the
fixed point. That strong-coupling problem is a natural continuation of this
paper, but it is not required for the cotunneling theorem proved above.

\bibliography{Main}

\end{document}